%% file: main.tex
\documentclass{amsart}

\usepackage[utf8x]{inputenc}
\usepackage{comment}
\usepackage{indentfirst}
\usepackage{amsthm}
\usepackage{amscd} 
\usepackage{enumitem}
\usepackage{graphicx}
\usepackage{amsmath}
\usepackage{amssymb}
\usepackage{xcolor}
\usepackage{tikz}
\usetikzlibrary{matrix,arrows}
\usepackage{tikz-cd}
\usepackage{mathtools}
\usepackage{mathalfa}
\usepackage[colorinlistoftodos]{todonotes}  
\usepackage{subfiles} 
\usepackage{upgreek} 
\usepackage{mathtools} 
\usepackage{adjustbox}
\usepackage{hyperref} 
\usepackage{cleveref}
\usepackage[mathscr]{eucal}
\usepackage[T1]{fontenc}

\usepackage{xr} 
\usepackage{graphicx} 
\graphicspath{ {./images/} }

\numberwithin{equation}{section}

\newtheorem{theorem}[equation]{Theorem}
\newtheorem*{theorem-non}{Theorem} 
\newtheorem{lemma}[equation]{Lemma}
\newtheorem{proposition}[equation]{Proposition}
\newtheorem{corollary}[equation]{Corollary}
\theoremstyle{definition}
\newtheorem{definition}[equation]{Definition}

\theoremstyle{remark}
\newtheorem{remark}[equation]{Remark}
\newtheorem{example}[equation]{Example}
\newtheorem{construction}[equation]{Construction}

\newcommand{\mcal}[1]{\mathcal{#1}}
\newcommand{\hcal}[1]{\hat{\mathcal{#1}}}
\newcommand{\mb}[1]{\mathbb{#1}}
\newcommand{\dhcal}[1]{\Sigma^{d-1}{\hcal{#1}}}

\newcommand{\dd}{\partial}
\newcommand{\Poincare}{Poincar\'e\ }
\newcommand{\mbcx}{\mb{C}^\times}
\newcommand{\cx}{\mbcx}
\newcommand{\icx}{I\mbcx}
\newcommand{\sus}{\Sigma^{\infty}_+}
\newcommand{\isorightarrow}{\xrightarrow{\sim}}

\newcommand{\mbD}{\mb{D}}

\newcommand{\oss}{\Omega^{\infty}}
\newcommand{\BordR}{Bord_d ^{\mcal{R}}}

\newcommand{\wmcal}[1]{\widetilde{\mcal{#1}}}

\newcommand{\E}[1]{\mb{E}_{#1}}
\newcommand{\Eone}{\mb{E}_{1}}
\newcommand{\Maps}{\underline{Maps}}

\title{Abelian Duality in Topological Field Theory}
\author{Yu Leon Liu}
\email{yuleonliu@math.harvard.edu}
\date{\today}

\begin{document}
\input{abstract}

\maketitle
\tableofcontents
\input{0_intro}

\input{1_pi_finite_spectra}

\input{2_finite_homotopy_tft}

\input{3_euler_tft}

\input{4_orientation_and_poincare_duality}

\input{5_bc_duality}

\input{6_abelian_duality}

\bibliographystyle{plain}
\bibliography{bib.bib}



\end{document}

%% file: abstract.tex
\begin{abstract}
We prove a topological version of abelian duality where the gauge groups 
are finite abelian. The theories are finite homotopy TFTs, 
topological analogues of the $p$-form $U(1)$ gauge theories. 
Using Brown-Comenetz duality, we extend the duality to finite homotopy TFTs of
$\pi$-finite spectra.
\end{abstract}

%% file: 0_intro.tex


\setcounter{section}{-1}
\section{Introduction} \label{sec0}

Abelian duality is a generalization of electromagnetic duality \cite{V95, W95} to 
higher gauge groups in general dimensions. 
The theories are $p$-form $U(1)$ gauge theories, 
free field theories whose dynamical fields 
are $p$-principal $U(1)$ bundles with connections \cite{Al85, Ell19, Ga88}. 
In $d$ dimension, abelian duality \cite{Bar95, Ell19, Kel09} 
is an equivalence of QFTs between $p$-form and $(d-p-2)$-form 
$U(1)$ gauge theories.

In this paper we study the case where the gauge groups are finite abelian. These theories are topological and can be formulated mathematically as topological field theories (TFT) \cite{At, Se88}.

Fix a dimension $d \geq 0$. An important class of topological field theories is the $d$ dimensional finite gauge theory \cite{DW, FQ} $Z_{BG}$, defined for any finite group $G$. Given a closed $d$ dimensional manifold $M$, partition function $Z_{BG}(M)$ counts the equivalence classes of principal $G$ bundles on $M$, weighted by automorphisms. 

For an abelian Lie group $A$, there are higher analogues of principal $A$-bundle, called $p$-principal $A$-bundle \cite{Br96, Bu12}. 
These bundles usually come with geometric structures such as a connection, 
and they can also have higher automorphisms. 
When $A$ is finite abelian, there are no geometric structures, and they can be understood via ordinary cohomology:
the equivalence class of $p$-principal $A$-bundle on $M$ is the cohomology group $H^p(M; A)$, the 
equivalence class of automorphisms of any principal $p$-bundle is $H^{p-1}(M;A)$, and automorphisms of any automorphism are $H^{p-2}(M;A)$...

For a finite abelian group $A$ and $p \in \mb{N}$, there exists a $d$ dimensional, 
$p$-form topological gauge theory $Z_{K(A,p)}$ with higher gauge group $A$ (\S \ref{sec2}).
Similar to the finite gauge theories, the partition function $Z_{K(A,p)}(M)$ counts equivalence classes of $p$-principal $A$-bundle on $M$, weighted by not only the automorphisms of these higher bundles, but also the automorphisms of the automorphisms and so on.

In QFT, abelian duality is a duality between $p$-form $U(1)$ gauge theories. In our topological case, the dual groups are Pontryagin dual groups. Let $A$ be a finite abelian group, the Pontryagin dual (character dual) group $\hat{A}$ is defined as $Hom(A, \cx)$. 
In \cite[\S 3.3]{FT}, Freed and Teleman proved that the $3$ dimensional finite gauge theories $Z_{K(A,1)}$ 
and $Z_{K(\hat{A},1)}$ are equivalent. In the same paper \cite[\S 9.1]{FT}, they noted that this duality
can be extended to general $p$ and arbitrary dimension $d$. We prove this duality:

\begin{theorem-non}[Abelian duality]\label{main-theorem0}
Let $A$ be a finite abelian group and $\hat{A}$ its Pontryagin dual. Let $d \geq 1$ be the dimension of our theories and $p \in \mathbb{Z}$. Let $Z_{K(A,p)}$, $Z_{K(\hat{A}, d-1-p)}$ be the $d$ dimensional finite homotopy TFTs associated to $K(A,n)$ and $K(\hat{A}, d-1-p)$. There is an equivalence of oriented topological field theories:
\begin{equation}
 Z_{K(A,p)} \simeq Z_{K(\hat{A}, d-1-p)} \otimes E_{|A|^{(-1)^p}},
\end{equation}
where is $E_{|A|^{(-1)^p}}$ is the $d$ dimensional Euler invertible TFT (\S \ref{sec3}) associated to $|A|^{(-1)^p} \in \cx$.
\end{theorem-non}

\begin{remark}
The two abelian duality field theories are equivalent up to twisting with 
the Euler TFT, which is an invertible field theory.
This was not stated in \cite{FT}.
\end{remark}




$p$-form $U(1)$ gauge theories are modeled by differential cohomology \cite{CS85, De71}, a differential refinement of ordinary integral cohomology. 
Spectra, or generalized cohomology theories, are interesting generalizations of ordinary cohomology theories. 
They, rather their differential analogues, also model fields in QFT and string theory. 
For example, Ramond-Ramond field in Type II superstring theory are modeled by differential K theory \cite{Fr, HS05}. 

The spectral generalization of finite abelian groups are $\pi$-finite spectra. 
A spectrum $\mcal{X}$ is $\pi$-finite (\S \ref{sec1}) if the stable homotopy groups $\pi_*(\mcal{X})$ are nontrivial in 
finitely many degrees, and each one is a finite abelian group. The canonical example of a $\pi$-finite spectrum is the  suspended Eilenberg-MacLane spectrum $\Sigma^p HA$, where $A$ is a finite abelian group. Given a $\pi$-finite spectrum $\mcal{X}$, we define a $d$ dimensional finite homotopy TFT $Z_{\mcal{X}}$  \cite{Fr1,FHLT,Qu, Tu}, which counts ``$\mcal{X}$'' bundles. Note that these TFTs can be defined for general $\pi$-finite spaces, however, we will work in the setting of $\pi$-finite spectra. See Remark \ref{sec2-remark} for how they are related.

Pontryagin duality can also be extended to $\pi$-finite spectra. In \cite{BC}, Brown and Comenetz defined a dual spectrum $\hcal{X}$ for any spectrum $\mcal{X}$. When $\mcal{X}$ is $\pi$-finite, then $\hcal{X}$ is also $\pi$-finite. We review this in \S \ref{sec5}. It is a generalization of Pontryagin duality:
let $A$ be an abelian group and $HA$ be its Eilenberg-MacLane spectrum, then 
\begin{equation}
\widehat{HA} \simeq H\hat{A},    
\end{equation}
where $\hat{A}$ is the Pontryagin dual group of $A$. 

In \cite[\S 9.1]{FT}, Freed and Teleman also pointed out that abelian duality 
can be extended to $\pi$-finite spectra, where Pontryagin duality is generalized to Brown-Comenetz duality. We prove that the two theories are equivalent up to twisting
with the Euler invertible TFT:

\begin{theorem-non}[Abelian duality]  \label{main_theorem'}
Let $\mcal{X}$ be a $\pi$-finite spectrum and $\hcal{X}$ its Brown-Comenetz dual. Let $Z_{\mcal{X}}, Z_{\dhcal{X}}$ be the corresponding $d$ dimensional finite homotopy TFTs. There is an equivalence of (suitably oriented) topological field theories:
\begin{equation}
    \mbD: Z_{\mcal{X}} \simeq Z_{\dhcal{X}} \otimes E_{|\mcal{X}|},
\end{equation}
where is $E_{|\mcal{X}|}$ is the $d$ dimensional Euler invertible TFT (\S \ref{sec3}).
\end{theorem-non}

\begin{remark}
Orientations are needed for abelian duality. When generalizing to $\pi$-finite spectra, we need a notion of orientation (\cite{ABGHR, Ru}) with respect to a ring spectrum. We review this in \S \ref{sec4}.
See Theorem \ref{sec6-theorem} for the rigorous statement.
\end{remark}

We would like to point out several aspects:
\begin{itemize}
    \item Path integrals are notoriously difficult to define mathematically. We work with finite abelian groups and $\pi$-finite spectra because there is mathematically well-defined theory of finite path integrals \cite{ DW, Fr1, FQ}. 
    \item Treatments of abelian duality focus on the dynamical part of the theory. Since the fields in our theories are rigid, we can focus on the topological aspects of abelian duality. We see that the Euler TFT, an invertible TFT with no dynamics, is present in the duality.
    \item The finite homotopy TFTs are not just interesting for their own right, they also encode (higher) symmetry for general (not necessarily topological) field theories 
        \cite{GKSW14}: $d$ dimensional theories with (non-anomalous) $G$ symmetry can be understood as boundary theories to the $d+1$ dimensional finite gauge theory $Z_{BG}$. 
        This even holds for lattice theories such as the Ising model \cite{FT}. 
        More generally, field theories with higher symmetries are boundary theories to $Z_{K(A,p)}$. 
        Therefore the abelian duality of $d+1$ dimensional theories gives duality of theories with higher $A$ symmetry and $\hat{A}$ symmetry. In 2 dimension, this is related to gauging/ungauging and the orbitfold construction \cite{DVVV89}.
    \item For a general $\pi$-finite spectrum $\mcal{X}$, a boundary field theory for the $d+1$ dimensional theory $Z_\mcal{X}$ is a theory with generalized symmetry \cite{GKSW14}. In this point of view, we see that our theorem gives duality of theories with generalized symmetry, if the symmetry is sufficiently commutative.
\end{itemize}
\subsection*{Outline}
In \S \ref{sec1} we review the basics of $\pi$-finite spectra.
In \S \ref{sec2} we define the finite homotopy TFTs associated to $\pi$-finite spectra.
In \S \ref{sec3} we define the Euler invertible TFT.
In \S \ref{sec4} we review orientation for generalized cohomology theories and \Poincare duality.
In \S \ref{sec5} we review Pontryagin and Brown-Comenetz dualities.
In \S \ref{sec6} we state and prove Theorem \ref{sec6-theorem}.

\subsection*{Acknowledgements}
I graciously thank my undergrad advisor, Dan Freed, 
for suggesting this problem and his continued guidance. 
I would like to thank David Ben-Zvi, Sanath Devalapurkar, 
Rok Gregoric, Mike Hopkins, Aaron Mazel-Gee, Riccardo Pedrotti, 
Wyatt Reeves, David Reutter, and Will Stewart for helpful conversations. 
I would like to especially thank Arun Debray for his many suggestions and feedback.
I would like to thank Meili Dubbs for her continuous support.


%% file: 1_pi_finite_spectra.tex
\section{\texorpdfstring{$\pi$}{pi}-finite spectra} \label{sec1}

In this section, we define $\pi$-finite spectra and their sizes. We assume some familiarity with spectra, see \cite{Ad74, HA}.


Let $S$, $Sp$ be the category of spaces and spectra. 
$\sus: S \rightarrow Sp$   
is the infinite suspension functor. Let $\tau_{\geq i}$ ($\tau_{\leq i}$)\ be the truncation functor that takes a spectrum to its $i$-th connective cover (truncation). Let $A$ be an abelian group, its Eilenberg-MacLane spectrum is denoted as $HA$.

Let $M$ be a (unpointed) space, $\mcal{X}$ a spectrum. We denote the homology and cohomology groups of $M$ with $\mcal{X}$ coefficients as $\mcal{X}_*(M)$, $\mcal{X}^*(M)$. Similarly, let $N$ be a pointed space, then $\wmcal{X}_*(N)$, $\wmcal{X}^*(N)$ are the reduced homology and cohomology groups of $N$ with $\mcal{X}$ coefficients. 

Given a cofiber sequence of (unpointed) spaces 
\begin{equation}
    N \rightarrow M \rightarrow M/N.
\end{equation}
$M/N$ is canonically pointed. The reduced (co)homology groups of $M/N$ are called the relative (co)homology groups of $(M,N)$, denoted as $\mcal{X}_*(M,N)$ ($\mcal{X}^*(M, N)$).

We denote the mapping spectrum $\Maps(\sus M, \mcal{X})$ as $\mcal{X}(M)$, its homotopy groups are cohomology groups:

\begin{equation}
\pi_i (\mcal{X}(M)) = \mcal{X}^{-i}(M).  
\end{equation}

A spectrum $\mcal{Y}$ is an extension of $\mcal{X}, \mcal{Z}$ if there is a fiber sequence 
\begin{equation}
    \mcal{X} \rightarrow \mcal{Y} \rightarrow \mcal{Z}.
\end{equation}

Given a full subcategory $C$, the extension closure $\overline{C}$ is the smallest full subcategory that contains $C$ and is closed under extension and suspension. For any property $P$, if all objects of $C$ satisfy $P$, and $P$ is closed under extension and suspension, then inductively, any object in the extension closure $\overline{C}$ satisfies $P$.

\begin{definition}\label{sec1-definition}
A spectrum $\mcal{X}$ is $\pi$-finite if
\begin{enumerate}
    \item $\pi_i \mcal{X}$ are non-trivial in only finitely many degrees.
    \item $\pi_i \mcal{X}$ are finite abelian groups for all $i \in \mathbb{Z}$.
\end{enumerate}
\end{definition}

We denote the category of $\pi$-finite spectra by $Sp^{fin}$. It is closed under suspension and extension. Canonical examples of $\pi$-finite spectra are $\Sigma^n HA$ where $A$ is a finite abelian group.

Given a $\pi$-finite spectrum $\mcal{X}$ with homotopy groups concentrated in degrees $\leq n$. Consider the fiber sequence
\begin{equation}
    \uptau_{\geq n} \mcal{X} \rightarrow \mcal{X} \rightarrow \uptau_{\leq n-1}\mcal{X}.
\end{equation}
As $\uptau_{\geq n} \mcal{X}$ only has nontrivial homotopy groups in degree $n$, it is a suspended finite Eilenberg-MacLane spectrum $\uptau_{\geq n} \mcal{X} \simeq \Sigma^n H\pi_n(\mcal{X})$. Doing this inductively on $n$, we see that $\mcal{X}$ lies in the extension closure of the full subcategory of finite Eilenberg-MacLane spectra.
Therefore, $Sp^{fin}$ is the extension closure of  finite Eilenberg-MacLane spectra. In fact, as finite abelian group are extensions of $\mathbb{F}_p$, $Sp^{fin}$ is the extension closure of finite Eilenberg-MacLane spectra of the form $H\mathbb{F}_p$.

Given a suitably finite graded abelian group, we have a notion of "homotopic size":
\begin{definition}
Let $A_\bullet = \bigoplus A_i$ be a $\mb{Z}$-graded abelian group. $A_\bullet$ is finite if all but finitely many $A_i$ are trivial and each $A_i$ is a finite abelian group. The size of $A_\bullet$ is
\begin{equation}
|A_\bullet| \coloneqq \prod_{i} |A_i|^{(-1)^i},   
\end{equation} where $|A_i|$ is the cardinality of $A_i$.
\end{definition}
We have a neat algebraic fact:
\begin{lemma}\label{sec1-lemma}
Given an exact sequence of finite graded abelian group $H^0_\bullet \rightarrow H^1_\bullet \rightarrow H^2_\bullet$, that is, a long exact sequence 
\begin{equation}\label{sec1-long}
    \cdots \rightarrow H_*^0 \rightarrow H_*^1 \rightarrow H_*^2 \rightarrow H_{*-1}^0 \rightarrow \cdots,
\end{equation}
then
\begin{equation}
|H^0_\bullet| |H^2_\bullet| = |H^1_\bullet|.    
\end{equation}
Alternatively, viewed the pieces of a long exact sequence as a finite graded abelian group $H_\bullet$, then
\begin{equation}
    |H_\bullet| = 1
\end{equation}
\end{lemma}
Let $\mcal{X}$ be a $\pi$-finite spectrum, then $\pi_\bullet(\mcal{X})$ is a finite graded abelian group. We can define its size:
\begin{definition}
The size of $\mcal{X}$, denoted as $|\mcal{X}|$, is the size of its homotopy groups 
\begin{equation}
   |\pi_\bullet(\mcal{X})| = \prod_i |\pi_i(X)|^{(-1)^i} = \cdots \frac{|\pi_0 \mcal{X}|}{|\pi_{-1} \mcal{X}|}  \frac{|\pi_2 \mcal{X}|}  {|\pi_1 \mcal{X}|} \cdots.  
\end{equation}
\end{definition}

\begin{proposition}\label{sec1-proposition1}
Given a fiber sequence $\mcal{X} \rightarrow \mcal{Y} \rightarrow \mcal{Z}$ of $\pi$-finite spectra, we have that $|\mcal{X}| \  |\mcal{Z}| = |\mcal{Y}|$.
\end{proposition}

\begin{proof}
This follows from applying Lemma \ref{sec1-lemma} to the long exact sequence of homotopy groups associated to the fiber sequence $\mcal{X} \rightarrow \mcal{Y} \rightarrow \mcal{Z}$.
\end{proof}

\begin{example}\label{sec1-example}
Let $\mcal{X}$ be a $\pi$-finite spectrum. We have a fiber sequence 
\begin{equation}
    \uptau_{\geq i}\mcal{X} \rightarrow \mcal{X} \rightarrow \uptau_{\leq i-1}\mcal{X}
\end{equation}
of $\pi$-finite spectra. By Proposition \ref{sec1-proposition1} we have 
\begin{equation}
    |\uptau_{\geq i}\mcal{X}|\ |\uptau_{\leq i-1}\mcal{X}| = |\mcal{X}|.
\end{equation}
\end{example}





\begin{proposition} \label{sec1-proposition2}
Let $M$ be a $d$ dimensional compact manifold (possibly with boundary) and $\mcal{X}$ a $\pi$-finite spectrum, then the mapping spectrum $\mcal{X}(M)$
is a $\pi$-finite spectrum of size 
\begin{equation} \label{sec1-eq}
   |\mcal{X}(M)| = |\mcal{X}|^{\chi (M)},
\end{equation}
where $\chi(M)$ is the Euler characteristic of $M$. 
\end{proposition}

\begin{proof}
Fix $M$,  we first proof that $\mcal{X}(M)$ is $\pi$-finite in the case that $\mcal{X} = HA$, $A$ finite abelian.
As $M$ is compact, $\pi_{-i}HA(M)\simeq H^i(M; A)$ are finite in each degree. It is also trivial when outside of degrees $0 \leq i \leq d$. Therefore $HA(M)$ is a $\pi$-finite spectrum. 

Consider $\mcal{X}(M)$ being $\pi$-finite as a property on $\mcal{X}$. This property is clearly closed under suspension. It is also closed under extension: given fiber sequence $\mcal{X} \rightarrow \mcal{Y} \rightarrow \mcal{Z}$, we have the following fiber sequence 
\begin{equation}
    \mcal{X}(M) \rightarrow \mcal{Y}(M) \rightarrow \mcal{Z}(M).
\end{equation}
If $\mcal{X}(M), \mcal{Z}(M)$ are $\pi$-finite, then so is $\mcal{Y}(M)$. Since $Sp^{fin}$ is the extension closure of finite Eilenberg-MacLane spectra, we see that $\mcal{X}(M)$ is $\pi$-finite for all $\mcal{X} \in Sp^{fin}$.

Now we prove that $|\mcal{X}(M)| = |\mcal{X}|^{\chi (M)}$. First consider the case that $\mcal{X} = H\mathbb{F}_p$, then 
\begin{equation}
    \pi_i(H\mathbb{F}_p(M)) = H^{-i}(M; \mathbb{F}_p)
\end{equation}
are finite dimensional $\mathbb{F}_p$ vector spaces. If $d$ be its dimension, then
\begin{equation}
    |\pi_i(H\mathbb{F}_p(M))| = p^d.
\end{equation}
As 
\begin{equation}
\sum_i (-1)^i dim_{\mathbb{F}_p} H^{i}(M; \mathbb{F}_p) = \chi(M)
\end{equation}
It follows that
\begin{equation}
|H\mathbb{F}_p(M)| = |H\mathbb{F}_p|^{\chi (M)}.
\end{equation}

Consider equation $|\mcal{X}(M)| = |\mcal{X}|^{\chi (M)}$ as a property on spectrum $\mcal{X}$. As both sides of the equation are multiplicative with respect to extensions, this property is closed under suspension and extension. Since $Sp^{fin}$ is the extension closure of spectra of the form $H\mathbb{F}_p$, it holds for all $\mcal{X} \in Sp^{fin}$.
\end{proof}

%% file: 2_finite_homotopy_tft.tex
\section{Finite Homotopy TFT} \label{sec2}
In this section we construct the $d$ dimensional (unoriented) TFT $Z_\mcal{X}$ associated to a $\pi$-finite spectrum $\mcal{X}$. Recall that a $d$ dimensional (unoriented) topological field theory (TFT) \cite{At, Lu09, Se88} is a symmetric monoidal functor
\begin{equation}
    Z: Bord_d \rightarrow Vect_{\mb{C}},
\end{equation}
where 
\begin{enumerate}
    \item $Bord_d$ is the $d$ dimensional bordism category. Its objects are closed $d-1$ dimensional manifolds, and morphisms are diffeomorphism classes of bordisms. It is symmetric monoidal under disjoint union.
    \item $Vect_\mb{C}$ is the category of finite dimensional $\mb{C}$-linear vector spaces, symmetric monoidal under tensor products.
\end{enumerate}

\begin{remark}
We are interested in TFTs with tangential structure, namely orientation in generalized cohomology. 
They are defined in Definition \ref{sec4-oriented-tft}.
\end{remark}

\begin{construction} \label{sec2-construction}
Fix a dimension $d \geq 1$, consider the following assignment: for any closed $d-1$ dimensional manifold $N$, we assign the vector space $\mathbb{C}[\mcal{X}^0(N)]$. Given a bordism $M: N \rightarrow N'$ with inclusions $p: N \subset M$ and $q: N' \subset M$, the map 
\begin{equation}
    Z_{\mcal{X}}(M):\mathbb{C}[\mcal{X}^0(N)] \rightarrow \mathbb{C}[\mcal{X}^0(N')] 
\end{equation}
is defined as follows: for $a \in \mcal{X}^0(N)$, considered as a basis element in $\mb{C}[\mcal{X}^0(N)]$, $Z_{\mcal{X}}(M)$ takes 
\begin{align}
a &\mapsto \frac{|\mcal{X}^{-1}(N')|}{|\mcal{X}^{-1}(M)|}
\frac{|\mcal{X}^{-2}(M)|}{|\mcal{X}^{-2}(N')|} \frac{|\mcal{X}^{-3}(N')|}{|\mcal{X}^{-3}(M)|}... \sum_{b \rightarrow a} q^*b \\
&= \frac{|\uptau_{\geq 1}\mcal{X}(M)|} {|\uptau_{\geq 1}\mcal{X}(N')|} \sum_{b \rightarrow a} q^*b\\
&= \frac{|\uptau_{\geq 1}\mcal{X}(M)|} {|\uptau_{\geq 1}\mcal{X}(N')|} \sum_{a'} \sum_{b \rightarrow a, b \rightarrow a'} a',
\end{align}
where $\sum_{b \rightarrow a}$ means sum over all $b \in \mcal{X}^0(M)$ such that $p^* b = a$.
\end{construction}

This defines a topological field theory:

\begin{proposition}
The assignment above defines a symmetric monoidal functor $Z_\mcal{X}: Bord_d \rightarrow Vect_\mb{C}$, that is, a topological field theory.
\end{proposition}

\begin{proof}
We first show that it is a functor, that is, the assignment $Z_\mcal{X}$ composes. Given two bordism $M: N \rightarrow N'$, $M': N' \rightarrow N''$, the composite is $M \sqcup_{N'} M': N \rightarrow N''$. 
We need to show that 
\begin{equation}
    Z_{\mcal{X}}(M') \circ Z_{\mcal{X}}(M) = Z_{\mcal{X}}(M \sqcup_{N'} M')
\end{equation}
Given $a \in \mcal{X}^0(N)$, considered as a basis vector on $\mb{C}[\mcal{X}^0(N)]$, then 
\begin{align}
Z_{\mcal{X}}(M') \circ Z_{\mcal{X}}(M) \ a &= Z_{\mcal{X}}(M') ( \frac{|\uptau_{\geq 1}\mcal{X}(M)|} {|\uptau_{\geq 1}\mcal{X}(N')|} \sum_{a'} \sum_{b \rightarrow a, b \rightarrow a'} a' ) \\ 
&= \frac{|\uptau_{\geq 1}\mcal{X}(M)|} {|\uptau_{\geq 1}\mcal{X}(N')|} \frac{|\uptau_{\geq 1}\mcal{X}(M')|} {|\uptau_{\geq 1}\mcal{X}(N'')|} \\ & \ \ \ \ \ \ \ \ \ \sum_{a'} \sum_{b \rightarrow a, b \rightarrow a'}
\sum_{a''} \sum_{b' \rightarrow a', b' \rightarrow a''} a''
\end{align}

On the other hand, 

\begin{equation}
    Z_{\mcal{X}}(M \sqcup_{N'} M') a = \frac{|\uptau_{\geq 1}\mcal{X}(M \sqcup_{N'} M')|} {|\uptau_{\geq 1}\mcal{X}(N'')|} \sum_{a''} \sum_{b \rightarrow a, b \rightarrow a'} a''
\end{equation}

Fix $a'' \in \mcal{X}(N'')$, it suffices to show that 

\begin{equation} \label{sec2-eq}
\frac{|\uptau_{\geq 1}\mcal{X}(M)| |\uptau_{\geq 1}\mcal{X}(M')|} {|\uptau_{\geq 1}\mcal{X}(N')|} \sum_{b,b', f^* b = g^* b'}   1 = |\uptau_{\geq 1}\mcal{X}(M \sqcup_{N'} M')| \sum_{c \rightarrow a, c \rightarrow a''} 1
\end{equation}

The sum $\sum_{b,b', f^* b = g^* b'}$ sums over pairs $b \in \mcal{X}^0(M)$, $b' \in \mcal{X}^0(M')$ that pulls back to the same element in $\mcal{X}^0(N')$. We have a (homotopy) pushout diagram:
\begin{equation} 
\begin{tikzcd}
N' \arrow[r, "f'"] \arrow[d, "g'"]
& M' \arrow[d, "g"] \\
M \arrow[r, "f"]
& M \sqcup_{N'} N'
\end{tikzcd}
\end{equation}
Consider the following truncated Mayer-Vietoris sequence: 
\begin{equation}\label{sec2-MV-sequence}
\begin{split}
\cdots \rightarrow \mcal{X}^{-1}(M \sqcup_{N'} M') \rightarrow \mcal{X}^{-1}(M) \oplus \mcal{X}^{-1}(M') \rightarrow \mcal{X}^{-1}(N) \\
\rightarrow \mcal{X}^0(M \sqcup_{N'} M') \rightarrow ker (\mcal{X}^{0}(M) \oplus \mcal{X}^{0}(M') \rightarrow \mcal{X}^{0}(N)) \rightarrow 0
\end{split}
\end{equation}

The elements of $ker (\mcal{X}^{0}(M) \oplus \mcal{X}^{0}(M') \rightarrow \mcal{X}^{0}(N))$ are exactly pairs $b, b', f^* b = g^* b'$. 
From the exactness at $ker (\mcal{X}^{0}(M) \oplus \mcal{X}^{0}(M') \rightarrow \mcal{X}^{0}(N))$, given $a, a''$,  there is a $c$ that pulls back to $a, a''$ iff there is $b, b', f^* b = g^* b'$ that pulls back to $a, a''$ respectively. Therefore for each $a, a''$, the right hand side of \ref{sec2-eq} vanishes iff the left hand side does.

In addition, for each $a, a''$, if the two sides doesn't vanish, we claim that the sum is independent of $a, a''$. For left hand side, this is because the set of $c \in \mcal{X}^0(M \sqcup_{N'} M')$ that maps to $a, a''$ is a torsor for the kernel of the map $\mcal{X}^0(M \sqcup_{N'} M') \rightarrow \mcal{X}^0(N) \oplus \mcal{X}^0(N'')$. The same is true for the left hand side.

Combining this two facts, it is sufficient to show the sum over all $a, a''$ of the two sides are equal:
\begin{align}
\frac{|\uptau_{\geq 1}\mcal{X}(M)| |\uptau_{\geq 1}\mcal{X}(M')|} {|\uptau_{\geq 1}\mcal{X}(N')|} 
|ker(\mcal{X}^{0}(M) \oplus \mcal{X}^{0}(M') \rightarrow \mcal{X}^{0}(N))| \\ = |\uptau_{\geq 1}\mcal{X}(M \sqcup_{N'} M')| |\mcal{X}^0(M \sqcup_{N'} M')|
\end{align}
which follows directly from applying Lemma \ref{sec1-lemma} to the exact sequence \ref{sec2-MV-sequence}.

Now we provide the symmetric monoidal data. On objects, it is given by the canonical isomorphisms
\begin{equation}
\begin{aligned}
\mb{C}[\mcal{X}^0(N \sqcup N')] &\simeq \mb{C}[\mcal{X}^0(N) \times \mcal{X}^0(N')] \\ 
&\simeq \mathbb{C}[\mcal{X}^0(N)] \otimes \mathbb{C}[\mcal{X}^0(N')]
\end{aligned}
\end{equation}
On morphisms, as the maps are natural, it follows that they are compatible with the isomorphism on objects.
\end{proof}

\begin{example}
Let $\mcal{X} = \Sigma^p HA$ where $A$ is a finite abelian group. Recall that  $H^p(M; A)$ classifies $p$-principal $A$-bundles on $A$, and $H^{p-i}(M; A)$ classifies level $i$ automorphism of the bundles and so on.

Given a closed $d$ manifold $M$, the partition function $Z_{\Sigma^p HA}(M)$ counts the number of $p$-principal $A$-bundles on $M$, 
weighted by automorphisms. For a closed $d-1$ manifold $N$, $Z_{\Sigma^p HA}(N)$, the space of states on $N$, is $\mb{C}[H^n(M; A)]$. 
This is the $p$-form gauge theory $Z_{K(A, p)}$ that we described in the introduction, where the higher gauge group is $A$.
\end{example}

\begin{remark}\label{sec2-remark}
Let $\oss: Sp \rightarrow S$ be the underlying space functor. The theory $Z_{\mcal{X}}$ depends only on the space $\oss \mcal{X}$, which is a $\pi$-finite space, that is, it has finitely many connected components, and each component has finite homotopy groups nontrivial in only finitely degrees.

In fact, finite homotopy TFTs $Z'_X$ \cite{Fr1,FHLT,Qu, Tu} can be defined for any $\pi$-finite space $X$. For a $\pi$-finite spectrum $\mcal{X}$, the TFT we defined $Z_{\mcal{X}}$ agrees with the more general version:
\begin{equation}
    Z_{\mcal{X}} \simeq Z'_{\oss \mcal{X}}.
\end{equation}
\end{remark}

%% file: 3_euler_tft.tex


\section{Euler TFT} \label{sec3}
In this section, we recall the $d$ dimensional Euler invertible TFT associated to a nonzero complex number $\lambda \in \cx$.

Recall that a TFT $Z$ is invertible if 
\begin{enumerate}
    \item for every closed $d-1$ manifold $N$, $Z(N)$ is a one dimensional vector space (a line).
    \item for every bordism $M: N \rightarrow N'$, $Z(M): Z(N) \rightarrow Z(N')$ is an isomorphism of lines.
\end{enumerate}

Let $M$ be a $d$ dimensional compact manifold (possibly with boundary), recall that its Euler characteristic $\chi(M)$ 
is defined as the alternating sum
\begin{equation}
    \sum^i (-1)^i dim_k H^i(M; k)
\end{equation}
for any field $k$. This is well-defined as the cohomology groups are finite dimensional $k$ vector spaces and nonzero only in degrees $0$ to $d$. 

\begin{definition}
Let $\lambda \in \mb{C}^\times$ be a nonzero complex number, the  $d$ dimensional Euler TFT $E_\lambda$ is defined as follows:
for any closed $d-1$ manifold $N$,
\begin{equation}
E_\lambda (N) \coloneqq \mb{C}.    
\end{equation}
For a bordism $M: N \rightarrow N'$, 
\begin{equation}
E_\lambda(M): \mb{C} \rightarrow \mb{C}     
\end{equation}
is defined to be multiplication by $\lambda^{\chi(M) - \chi(N)} \in \cx$.
\end{definition}

It is a topological field theory by the following lemma:

\begin{lemma}
Given closed $d-1$ manifolds $N$, $N'$, $N''$ and bordisms $M: N \rightarrow N'$ and $M': N' \rightarrow N''$, then 
\begin{equation}
\chi(M \sqcup_{N'} M') - \chi(N) = \chi(M) - \chi(N) + \chi(M') - \chi(N').    
\end{equation}
\end{lemma}

\begin{remark}
The Euler TFTs $E_{\lambda}$ are invertible field theories.
\end{remark}

\begin{example}
Let $\lambda \neq 1, -1$. In even dimensions $d = 2n$, the $d$ dimensional sphere $S^d$ has Euler characteristic $\chi(S^d) = 2$. Therefore $E_\lambda(S^d) = \lambda^2 \neq 1$. We see that in even dimensions, the Euler TFT $E_\lambda$ is nontrivial.
\end{example}

In odd dimensions, by \Poincare duality (with $\mb{F}_2$ coefficients), the Euler characteristic of a closed $d$ manifold is $0$, so $Z(M) = 1$ for every closed $d$ manifold $M$. In fact, we have a stronger statement: 
\begin{proposition}\label{sec3-prop}
Let $d$ be odd. For any $\lambda \in \mathbb{C}^\times$, the $d$ dimensional Euler TFT is trivial, that is, $E_\lambda \simeq Z_{triv}$.
\end{proposition}

\begin{proof}
To show that $E_\lambda \simeq Z_{triv}$, we have to give a  natural isomorphism $\alpha: Z_{triv} \isorightarrow E_\lambda $ between the two functors.
We need the following data:
For every closed d-1 manifold $N$, we need  
\begin{equation}
\alpha(N): Z_{triv} (N)= \mb{C} \isorightarrow \mb{C} = E_\lambda(N),    
\end{equation}
which sends $1 \in \mb{C} = Z_{triv} (N)$ to a nonzero element 
\begin{equation}
\alpha_N \coloneqq \alpha(N)(1).     
\end{equation}
$\alpha$ also needs to satisfy the following compatibility condition: given a bordism $M: N \rightarrow N'$, we need a commutative diagram
\begin{equation}
    \begin{tikzcd}
Z_{triv}(N) \arrow[r, "Z_{triv}(M)"] \arrow[d, "\alpha(N)"]
& Z_{triv}(N')  \arrow[d, "\alpha(N')" ] \\
E_\lambda (N) \arrow[r, "E_{\lambda}(M)" ]
& E_\lambda(N').
\end{tikzcd}
\end{equation}
Tracking where  
\begin{equation}
1 \in \mb{C} = Z_{triv}(N)    
\end{equation}
goes, we see that we need to show that 
\begin{equation}\label{4.2.equation}
\alpha_{N'} = \lambda^{\chi(M) - \chi(N)} \alpha_N.    
\end{equation}
We claim that for 
\begin{equation}
\alpha_N = \lambda^{\frac{1}{2}\chi(N)},    
\end{equation}
Equation \ref{4.2.equation} is satisfied.
This is equivalent to showing that 
\begin{equation}
\chi(M) = \frac{1}{2}(\chi(N) + \chi(N')) = \frac{1}{2}\chi(\dd M).    
\end{equation}
Let $k = \mb{F}_2$. As every manifold is $k$-oriented, we have Poincare duality:
\begin{equation}
H^*(M; k) \simeq H_{d-*}(M, \dd M; k).    
\end{equation}
As $d$ is odd, we see that 
\begin{align}
\chi(M) &= \chi(H^*(M;k))\\ 
        &= \chi(H_{d-*}(M, \dd M; k))\\ 
        &= -\chi(H_*(M, \dd M; k)\\ 
        &= -\chi(M, \dd M),
\end{align}
Finally, consider the long exact sequence associated to the cofiber sequence $\dd M \rightarrow M \rightarrow M/\dd M$:
\begin{align*}
    \cdots \rightarrow H^*(M, \dd M; k) \rightarrow H^*(M;k) \rightarrow H^*(\dd M;k ) \rightarrow \cdots.
\end{align*}
By the additive version of Lemma \ref{sec1-lemma}, we see that 
\begin{align}
\chi(M) &= \chi(M, \dd M) + \chi(\dd M)\\  &= -\chi(M) + \chi(\dd M).    
\end{align}
Thus 
\begin{equation}
\chi(M) = \frac{1}{2} \chi(\dd M).      
\end{equation}
\end{proof}

%% file: 4_orientation_and_poincare_duality.tex


\section{Orientation and \Poincare Duality for Spectra} \label{sec4}
In this section we review the general theory of orientation in generalized cohomology theories \cite{ABGHR, Ru}. 
The new ingredient here is a $\Eone$-ring spectrum \cite{HA}, which is a homotopically associative ring.
Let $M$ be a $d$ manifold, $\mcal{R}$ a $\E{1}$-ring spectrum, its homotopy groups $\pi_*\mcal{R}$ inherits a graded ring structure. 
\begin{definition}
An $\mcal{R}$-orientation on $M$ is a homology class 
\begin{equation}
[M] \in \mcal{R}_d(M, \dd M)     
\end{equation}
satisfying the following condition: for every interior point $x \in M^o$, the image of $[M]$ under
\begin{equation}
\mcal{R}_d(M, \dd M) \rightarrow \mcal{R}_d(M, M-x) \simeq \wmcal{X}_d(S^d) \simeq \pi_{0}\mcal{R}    
\end{equation}
is an multiplicative unit in the ring $\pi_* \mcal{R}$.
\end{definition}

An $\mcal{R}$-orientation on a $d$-dimensional manifold gives $\mcal{R}$-orientation on the boundary:

\begin{proposition} \cite{May, Ru}
Let $M$ be a $d$ manifold. A $\mcal{R}$-orientation on $M$, $[M] \in \mcal{R}_d(M, \dd M)$, gives a class $\dd [M] \in \mcal{R}_{d-1}(N)$ via the natural boundary map 
\begin{equation}
\dd : \mcal{R}_*(M, \dd M) \rightarrow \mcal{R}_{*-1}(\dd M).    
\end{equation}
The class $\dd [M] \in \mcal{R}_{*-1}(\dd M)$ is a $\mcal{R}$-orientation on the boundary $\dd M$.
\end{proposition}


\begin{definition} \label{sec4-oriented-tft}
For any $\mb{E}_1$-ring spectrum $\mcal{R}$, the $d$ dimensional $\mcal{R}$-oriented bordism category $Bord_d ^{\mcal{R}}$ is defined as follows:
\begin{enumerate}
    \item the objects are closed $\mcal{R}$-oriented $d-1$ manifold $(N, [N])$.
    \item a morphism \begin{equation}
(M, [M]): (N, [N]) \rightarrow (N', [N'])    
\end{equation} is a (diffeomorphism class of) $\mcal{R}$-oriented bordism $M$ with the orientation of the boundary $\dd [M]$ restricts to $[N]$ on $N$ and $-[N']$ on $N'$.
\end{enumerate}
It is symmetric monoidal under disjoint union. See \cite{Lu09} for details.
\end{definition}
\begin{remark}
The minus sign is needed for $\mcal{R}$-oriented bordisms to compose.
\end{remark}
Now we can define $\mcal{R}$-oriented TFTs:
\begin{definition}
A $\mcal{R}$-oriented topological field theory is a symmetric monoidal functor 
\begin{equation}
    Z: Bord_d ^{\mcal{R}} \rightarrow Vect_{\mb{C}}.
\end{equation}
\end{definition}
\begin{remark}
The forgetful functor
\begin{equation}
Bord_d ^{\mcal{R}} \rightarrow Bord_d
\end{equation}
is a symmetric monoidal functor, thus any (unoriented) TFT pullback to a $\mcal{R}$ oriented TFT.
\end{remark}

Orientation gives \Poincare duality via the cap product construction:

\begin{construction}
Let $\mcal{R}$ be a $\mb{E}_1$ ring spectrum, and $\mcal{X}$ a left $\mcal{R}$ module. Let 
\begin{equation}
f: N \rightarrow N' \wedge N''    
\end{equation}
be a map of pointed spaces. 
The cap product \cite{Ad74} is a map
\begin{equation}\label{cap}
-\smallfrown -: \wmcal{R}_m(N) \otimes \wmcal{X}^n(N') \rightarrow \wmcal{X}_{m-n}(N'').
\end{equation}
\end{construction}
In our setting, let $M$ be a $\mcal{R}$-oriented $d$-manifold with boundary $\dd M = N \sqcup N'$. We have the orientation class $[M] \in \mcal{R}_d(M, \dd M)$. 
Consider the map 
\begin{equation}
    M/{\dd M} \rightarrow M/N \wedge M/N'
\end{equation}
of pointed spaces. From Equation \ref{cap}  we get a map 
\begin{equation}
[M] \smallfrown - :  \mcal{X}^*(M, N) \rightarrow \mcal{X}_{d-*}(M, N').
\end{equation}
We denote this map by $\int_{[M, N]}$.
In the case that $N = \varnothing$, then we will denote this as $\int_{[M]}$. \Poincare duality states that taking cap products give functorial isomorphisms:

\begin{theorem}\label{poincare} \cite{May, Ru}
Let $\mcal{R}$ be a $\E{1}$-ring spectrum and $\mcal{X}$ a left $\mcal{R}$-module spectrum. Let $M$ be a $\mcal{R}$-oriented $d$-manifold with boundary $\dd M = N \sqcup N'$. We denote the orientation class as $[M]$. It restricts to orientations $[N], [N']$ on the boundaries. Cap product gives isomorphisms of long exact sequences:
\begin{equation}
\begin{tikzcd}
\cdots \arrow{r} \arrow{d}
    & \mcal{X}^{*}(M, N')\arrow{d}{\int_{[M,N']}} \arrow{r}
        & \mcal{X}^{*}(M) \arrow{d}{\int_{[M]}} \arrow{r}{} 
            & \mcal{X}^{*}(N') \arrow{d}{\int_{[N']}} \arrow{r}
                & \cdots \arrow{d}\\
\cdots \arrow{r}
    & \mcal{X}_{d-*}(M, N) \arrow{r}{}
        & \mcal{X}_{d-*}(M, \dd M) \arrow{r}{}
            & \mcal{X}_{d-1-*}(N') \arrow{r}
                & \cdots
\end{tikzcd}
\end{equation}
\end{theorem}

\begin{example}
Let $\mcal{R} = H\mathbb{F}_2$. There is only one unit in $\pi_0(H\mb{F}_2) = \mb{F}_2$, 
the unique local choice glue together to give a $H\mb{F}_2$-orientation. 
Therefore any manifold is uniquely $H\mb{F}_2$-oriented. The $H\mb{F}_2$ bordism category is simply the unoriented bordism category $Bord_d$. $\pi$-finite $H\mb{F}_2$-module spectra can be represented as bounded chain complexes of finite dimensional $\mb{F}_2$-vector spaces.
\end{example}

\begin{example}
Let $\mcal{R} = H\mb{Z}$, then $H\mb{Z}$-orientation is the standard notion of orientation, 
and the $H\mb{Z}$-oriented bordism category is the oriented bordism category. 
$\pi$-finite $H\mb{Z}$-module spectra can be represented as bounded chain complexes of finite abelian groups.
\end{example}

\begin{example}
Let $\mcal{R} = \mcal{S}$ the sphere spectrum, then a $\mcal{S}$-orientation is a stable framing, that is, a trivialization of the stable normal bundle \cite{Ru}.  $\pi$-finite $\mcal{S}$ modules are simply $\pi$-finite spectra.
\end{example}


%% file: 5_bc_duality.tex


\section{Pontryagin and Brown-Comenetz Dualities} \label{sec5}
In this section we review Pontryagin duality for finite abelian groups and Brown-Comenetz duality \cite{BC} for $\pi$-finite spectra. We first start with Pontryagin duality:
Let $A$ be a abelian group, $\cx$ the group of nonzero complex number. The Pontryagin dual group $\hat{A}$ is the group of characters $Hom(A, \cx)$. Let $Ab$ be the category of abelian group, taking Pontryagin dual defines a contravariant functor:

\begin{equation}
\hat{D} \coloneqq Hom(-, \cx): Ab \rightarrow (Ab)^{op}.    
\end{equation}

The canonical pairing
\begin{equation}
    \mu: A \otimes \hat{A} \rightarrow \cx.
\end{equation}
gives rise to a map $\alpha_\mu: \mb{C}[A] \otimes \mb{C}[\hat{A}] \rightarrow \mb{C}$. On basis vectors it sends $a \otimes \alpha \rightarrow \mu(a,\alpha) = \alpha(a)$, where we identify an element $a \in A$ with the standard basis vector in $\mb{C}[A]$. 

Let $Ab^{fin}$ be the full subcategory of finite abelian group. We now assume that $A$ is a finite abelian group. In this case, 
the pairing is a nondegenerate pairing, that is, the induced map $(\mb{C}[A])^* \rightarrow \mb{C}[\hat{A}]$ is an isomorphism. Taking dimension, this implies that the dual groups have equal cardinality: $|A| = |\hat{A}|$. We see that Pontryagin dual restricts to a functor:
\begin{equation}
\hat{D}: Ab^{fin} \rightarrow ((Ab)^{fin})^{op}.    
\end{equation}
In addition, we have the corollary:
\begin{corollary} \label{sec5-corollary}
The map 
\begin{equation}
\begin{aligned}
    \mb{C}[A] &\rightarrow \mb{C}[\hat{A}] \\
           a &\mapsto \sum_\alpha \mu(a,\alpha)\  \alpha
\end{aligned}
\end{equation}
is an isomorphism.
\end{corollary}
The nondegeneracy of the pairing implies that taking Pontryagin dual  $\hat{D}$ is an involution on $Ab^{fin}$:
\begin{theorem}\label{sec5-theorem}
Restricted to $Ab^{fin}$,  $\hat{D}^2 \simeq id$.
\end{theorem}
As a corollary, we get that Pontryagin duality is in fact a duality on finite abelian groups:
\begin{corollary}
$\hat{D}: Ab^{fin} \rightarrow (Ab^{fin})^{op}$ is an equivalence of categories.
\end{corollary}

\begin{remark}
The Pontryagin dual of $A$ is commonly defined as $Hom(A, \mb{Q}/\mb{Z})$. There is a natural map 
\begin{equation}
\mb{Q}/\mb{Z} \xrightarrow{exp(2\pi i)} \cx.
\end{equation}
For a finite abelian group $A$, the map
\begin{equation}
    Hom(A, \mb{Q}/\mb{Z}) \rightarrow Hom(A, \cx)
\end{equation}
is an isomorphism and the two notions coincide.
We choose $\cx$ over $\mb{Q}/\mb{Z}$ as our TFTs are complex-valued.
\end{remark}

We now generalize Pontryagin duality to $\pi$-finite spectra. There is a spectrum $\icx$ that plays the role of $\cx$:

\begin{theorem} \cite{BC}
There exists a spectrum $\icx$ with the following data: for any spectra $\mcal{X}$, there is a functorial equivalence
\begin{equation}\label{sec5-equation}
\pi_{-*}(\Maps(\mcal{X}, \icx)) \simeq \widehat{\pi_*(\mcal{X})}.
\end{equation}
More precisely, we view both sides as families of contravariant functors $Sp \rightarrow Ab$, and we claim that there is a natural isomorphism between these two families of functors, compatible with the connecting homomorphisms.
\end{theorem}

\begin{definition}
Let $\mcal{X}$ be a spectrum, the Brown-Comenetz dual spectrum $\hcal{X}$ is defined to be the mapping spectrum $\Maps(\mcal{X}, \icx)$.
\end{definition}
This defines a contravariant functor 
\begin{equation}
    \hcal{D} \coloneqq \Maps(-, \icx): Sp \rightarrow Sp^{op}.
\end{equation}

\begin{example}
Let $\mcal{X}$ be the sphere spectrum $\mcal{S}$. Then 
\begin{equation}
\hcal{S} = \Maps(\mcal{S},\icx) \simeq \icx.   
\end{equation}
Therefore $\icx$ is the Brown-Comenetz dual of the sphere spectrum $\mcal{S}$. This is similar to the fact that $\cx$ is the Pontryagin dual group of $\mb{Z}$.
\end{example}
\begin{remark}
The common approach to Brown-Comenetz uses a similarly defined $I\mb{Q}/\mb{Z}$ rather than $I\cx$. As with the abelian group case, they give the same answers on $\pi$-finite spectra. We use $I\cx$ over $I\mb{Q}/\mb{Z}$ because the target of our TFTs are complex-valued and $\icx$ is the natural target for invertible TFTs (see \cite{FH}).
\end{remark}

By equation \ref{sec5-equation} and the fact that the Pontryagin dual of a finite abelian group is also finite, we see that taking Brown-Comenetz dual restricts to a functor
\begin{equation}
    \hcal{D} \coloneqq \Maps(-, \icx): Sp^{fin} \rightarrow (Sp^{fin})^{op}.
\end{equation}

There is a natural transformation $id \rightarrow \hcal{D}^2$, given by 
\begin{equation} \label{sec5-eq}
\begin{aligned}
\mcal{X} &\rightarrow \hat{\hcal{X}} = \Maps(\Maps(\mcal{X}, \icx), \icx) \\
      x &\mapsto (\alpha \mapsto \alpha(a)).
\end{aligned}
\end{equation}

Restricting to $\pi$-finite spectra, this natural transformation is an isomorphism:
\begin{theorem}\label{sec5-theorem2}\cite{BC}
For $\pi$-finite spectrum $\mcal{X}$, the natural map \ref{sec5-eq} is an isomorphism. Therefore, we have an equivalence of functors 
\begin{equation}
\hcal{D}^2 \simeq id: Sp^{fin} \rightarrow Sp^{fin}.    
\end{equation}
\end{theorem}
As a corollary, we get that Brown-Comenetz duality is in fact a duality of $\pi$-finite spectra:
\begin{corollary}
$\hat{D}: Sp^{fin} \rightarrow (Sp^{fin})^{op}$ is an equivalence.
\end{corollary}
We also have the following corollary:
\begin{corollary}\label{sec5-corollary2}
Let $\mcal{X}$ be a $\pi$-finite spectra and $N \rightarrow M \rightarrow M/N$ a cofiber sequence of finite $CW$ complexes. The Pontryagin dual of the long exact sequence of cohomology group with $\mcal{X}$ coefficients: 
\begin{equation}
   \cdots \rightarrow     \mcal{X}^*(M,N) \rightarrow \mcal{X}^*(M) \rightarrow \mcal{X}^*(N) \rightarrow \cdots
\end{equation}
is canonically isomorphic to  the long exact sequence of homology group with $\hcal{X}$ coefficients:
\begin{align}
  \cdots \leftarrow \hcal{X}_*(M,N) \leftarrow \hcal{X}_*(M) \leftarrow \hcal{X}_*(N) \leftarrow \cdots.
\end{align}
\end{corollary}


%% file: 6_abelian_duality.tex
\section{Abelian Duality}\label{sec6}

Fix $d \geq 1$ the dimension of our theory. 
Let $\mcal{R}$ be a $\E{1}$-ring spectrum and $\mcal{X}$ a $\pi$-finite left $\mcal{R}$-module spectrum, its size (\S 2) is denoted as $|\mcal{X}|$. 
The Brown-Comenetz dual $\hcal{X}$ is a $\pi$-finite right $\mcal{R}$-module. 
In \S \ref{sec2}, we defined the $d$-dimensional finite homotopy TFTs $Z_{\mcal{X}}$ and $Z_{\dhcal{X}}$ associated to $\mcal{X}$ and $\dhcal{X}$. In addition, if $\lambda$ is a nonzero complex number, we have the $d$-dimensional Euler invertible TFT $E_\lambda$ (\S \ref{sec3}). 

In \S \ref{sec4} we defined the $\mcal{R}$-oriented bordism category $\BordR$. Any unoriented TFT can be viewed as a $\mcal{R}$-oriented TFT by precomposing with the forgetful map $\BordR \rightarrow Bord_d$.
We view $Z_{\mcal{X}},Z_{\dhcal{X}}, E_{|\mcal{X}|}$ as $\mcal{R}$-oriented theories. 
\begin{theorem}[Abelian duality]\label{sec6-theorem}
There is an equivalence of $\mcal{R}$-oriented TFTs 
\begin{equation}
    \mb{D}: Z_{\mcal{X}} \simeq Z_{\dhcal{X}} \otimes E_{|\mcal{X}|}.
\end{equation}
\end{theorem}
\begin{remark}
In general, $Z_{\mcal{X}}$ and $Z_{\dhcal{X}} \otimes E_{|\mcal{X}|}$ are not equivalent as unoriented theories, even though both sides can be extended to unoriented theories. This is because we need to use \Poincare duality in an essential way. For example, let $d=2$, let $\mcal{X} = \Sigma H\mb{F}_3$, $\Sigma^{d-1} \hcal{X} \simeq H\mb{F}_3$. Consider the theories on the Klein bottle $K$. Note $K$ is not $H\mb{F}_3$-orientable. If Theorem \ref{sec6-theorem} holds, then 
\begin{equation}
    Z_{\Sigma H\mb{F}_3}(K) = Z_{H\mb{F}_3}(K) \cdot E_{|\mb{F}_3|^{-1}}(K).
\end{equation}
On one hand,
\begin{equation}
Z_{\Sigma H\mb{F}_3}(K) = |\tau_{\geq -1}  H\mb{F}_3(k)| = \frac{|H^1(K;\mb{F}_3)|}{|H^0(K;\mb{F}_3)|} = \frac{|\mb{F}_3|}{|\mb{F}_3|} = 1.    
\end{equation}
On the other hand,
\begin{equation}
    Z_{H\mb{F}_3}(K) = |H^0(K;\mb{F}_3)| = 3,
\end{equation}
and 
\begin{equation}
    E_{|\mb{F}_3|^{-1}}(K) = 3^{-\chi(K)} = 1.
\end{equation}
We see that they are different.
\end{remark}


Here's some consequences of Theorem \ref{sec6-theorem}:

\begin{corollary}\label{sec6-corollary}
When $d$ is odd, we have an equivalence of $\mcal{R}$-oriented TFTs
\begin{equation}
 Z_{\mcal{X}} \simeq Z_{\dhcal{X}}.
\end{equation}
\end{corollary}

\begin{proof}
By Proposition \ref{sec3-prop}, when $d$ is odd, $E_\lambda$ is isomorphic to the trivial theory. Therefore $Z_{\mcal{X}} \simeq Z_{\dhcal{X}}$.
\end{proof}

Recall that $H\mb{Z}$ orientation is the same as the classical notion of orientation on manifolds. 
Thus if $\mcal{X}$ is a $\pi$-finite $H\mb{Z}$-module, then we have an equivalence of oriented theories. 
Apply Theorem \ref{sec6-theorem} to $\pi$-finite $H\mb{Z}$-module $\Sigma^p HA$, we get:

\begin{corollary} \label{K(A,n)}
Let $A$ be a finite abelian group and $\hat{A}$ the Pontryagin dual. In $d$ dimension, we have an equivalence of $d$ dimensional oriented TFTs:
\begin{equation}
 Z_{K(A,p)} \simeq Z_{K(\hat{A}, d-1-p)} \otimes E_{|A|^{(-1)^p}}.
\end{equation}
\end{corollary}

The rest of the paper is about proving Theorem \ref{sec6-theorem}:
\begin{proof}[Proof of Theorem \ref{sec6-theorem}]
From now on, all manifolds, and bordisms are $\mcal{R}$-oriented. We will suppress the $\mcal{R}$-orientation notations.

To give an equivalence, we need to define an isomorphism of states 
\begin{equation}
    Z_\mcal{X} (N) \rightarrow Z_{\dhcal{X}}(N) \otimes E_{|\mcal{X}|}(N),
\end{equation}
for closed $d-1$ manifolds $N$, 
and check that it is compatible with bordisms. As $E_{|\mcal{X}|}(N) = \mb{C}$, it is sufficient to give maps 
\begin{equation}
    \mbD(N): Z_\mcal{X} (N) \rightarrow Z_{\dhcal{X}}(N).
\end{equation}
This is done as follows:
\begin{construction} \label{sec6-construction}
By Pontryagin duality, there is a pairing 
\begin{equation}
ev_N(-,-): \mcal{X}^*(N) \times \hcal{X}_*(N) \rightarrow \mb{C}^\times.
\end{equation}
Note that this exist for any topological space $N$.
As $N$ is a (compact) manifold, the homology and cohomology groups are finite. This pairing is exhibits $\mcal{X}^*(N)$ and $\hcal{X}_*(N)$ as Pontryagin dual of each other.
Compose this with the \Poincare duality isomorphism \ref{poincare}: 
\begin{equation}
\int_{[N]}: \hcal{X}^{d-1-*}(N) \isorightarrow  \hcal{X}_{*}(N),
\end{equation}
we get a pairing
\begin{align}
\mcal{X}^*(N) \times \hcal{X}^{d-1-*}(N) &\rightarrow \mb{C}^\times\\
(a,\alpha) &\mapsto ev_N(a, \int_{[N]}\alpha) \label{6.1.0}
\end{align}
When $* = 0$, we denote this pairing as 
\begin{equation} \label{pairing}
\langle -,- \rangle_N: \mcal{X}^0(N) \times \hcal{X}^{d-1}(N) \rightarrow \mb{C}^\times.
\end{equation} 
It exhibits $\mcal{X}^0(N)$ and $\hcal{X}^{d-1}(N)$ as the Pontryagin dual of each other.
Note that this depends on the orientation class of $N$, reversing the orientation inverts this pairing.

Recall that 
\begin{equation}
Z_\mcal{X}(N) = \mb{C}[\mcal{X}^0(N)]
\end{equation} 
and
\begin{equation}
Z_{\dhcal{X}}(N) = \mb{C}[\hcal{X}^{d-1}(N)].
\end{equation}
We will denote elements of $\mcal{X}^0(N)$ as $a$, and $\hcal{X}^{d-1}(N)$ as $\alpha$, and view them as basis vectors for $Z_\mcal{X}(N)$ and $Z_{\dhcal{X}}(N)$ respectively. Now we can define the isomorphism on states:  
\begin{equation}
\begin{aligned} \label{iso_on_state}
\mb{D}(N): \mb{C}[\mcal{X}^0(N)] &\rightarrow \mb{C}[\hcal{X}^{d-1}(N)] \\
a &\mapsto |\tau_{\geq 1} \mcal{X}(N)| \ \sum_{\alpha} \langle a, \alpha \rangle_N\  \alpha.
\end{aligned}
\end{equation}
This is an isomorphism of vector spaces by Corollary \ref{sec5-corollary}.
\end{construction}

It remains to show that this intertwines with bordisms.
Given $M: N \rightarrow N'$ in $Bord_d$, with the inclusion maps $p: N \hookrightarrow M$ and $q: N' \hookrightarrow M$. We have to show that the following diagram commute:

\begin{equation} \label{5.3.comm}
    \begin{tikzcd}[column sep = 10em]
Z_\mcal{X}(N) \arrow[r, "Z_\mcal{X}(M)"] \arrow[d, "\mb{D}(N)"]
& Z_\mcal{X}(N') \arrow[d, "\mb{D}(N')" ] \\
Z_{\dhcal{X}}(N) \arrow[r, "Z_{\dhcal{X}}(M) * |\mcal{X}|^{\chi(M) - \chi(N)}"' ]
& Z_{\dhcal{X}}(N')
\end{tikzcd}
\end{equation}

Note that we have canonically identified 
\begin{equation}
Z_{\dhcal{X}}(N) \otimes E_{|\mcal{X}|}(N) \simeq Z_{\dhcal{X}}(N).
\end{equation}
The factor 
\begin{equation}
|\mcal{X}|^{\chi(M) - \chi(N)}
\end{equation}
in the bottom arrow comes from
\begin{equation}
E_{|\mcal{X}|}(M): E_{|\mcal{X}|}(N)=\mb{C} \rightarrow \mb{C} = E_{|\mcal{X}|}(N').
\end{equation}

We will prove that diagram \ref{5.3.comm} commutes in two lemmas:
\begin{lemma}\label{lemma_1}
$\mb{D}(N') \circ Z_{\mcal{X}}(M)$ and $ Z_{\dhcal{X}}(M) \circ \mbD(N)$ differ by a constant $\lambda(M)$.
\end{lemma}

\begin{lemma}\label{lemma_2}
$\lambda(M) = |\mcal{X}|^{\chi(M) - \chi(N)}$.
\end{lemma}
The two lemmas are proven in \S \ref{sec6-1} and \S \ref{sec6-2} respectively.
\end{proof}

\subsection{Proof of Lemma 1}\label{sec6-1}

We borrow the notation from above. This section is devoted to proving Lemma \ref{lemma_1}, which we repeat here: 
\begin{lemma}
$\mb{D}(N') \circ Z_{\mcal{X}}(M)$ and $ Z_{\dhcal{X}}(M) \circ \mbD(N)$ differ by a constant $\lambda(M)$.
\end{lemma}

\begin{proof}

From now on, we will denote elements of 
\begin{equation}
\mcal{X}^0(N),\  \mcal{X}^0(M),\  \mcal{X}^0(N')     
\end{equation}
as $a, b,$ and $a'$. Similarly, we will denote elements of 
\begin{equation}
\hcal{X}^{d-1}(N),\  \hcal{X}^{d-1}(M), \ \hcal{X}^{d-1}(N')     
\end{equation}
as $\alpha, \beta,$ and $\alpha'$. We also use the summing convention that $\sum_{b}$ means summing over all $b\in \mcal{X}^0(M)$, and $\sum_{b \rightarrow a}$ means summing over all $b \in \mcal{X}^0(M)$ such that $p^*(b) = a$.

We denote the inclusion maps $p: N \hookrightarrow M$ and $q: N' \hookrightarrow M$. We have pullback maps 
\begin{equation}
p^*: \mcal{X}^0(M) \rightarrow \mcal{X}^0(N),\ \ q^*: \mcal{X}^0 (M) \rightarrow \mcal{X}^0 (N').
\end{equation}
Similarly we have
\begin{equation}
\hat{p}^*: \hcal{X}^{d-1}(M) \rightarrow \hcal{X}^{d-1}(N), \ \ \hat{q}^*: \hcal{X}^{d-1}(M) \rightarrow \hcal{X}^{d-1}(N').
\end{equation} 

First we will calculate $\mb{D}(N') \circ Z_{\mcal{X}}(M)$. By Construction \ref{sec2-construction}, $Z_{\mcal{X}}(M)$ sends 
\begin{align}
a &\mapsto 
\frac{|\uptau_{\geq 1}\mcal{X}(M)|} {|\uptau_{\geq 1}\mcal{X}(N')|} \sum_{b \rightarrow a} q^*b\\
&= \frac{|\uptau_{\geq 1}\mcal{X}(M)|} {|\uptau_{\geq 1}\mcal{X}(N')|} \sum_{a'} \sum_{b \rightarrow a, b \rightarrow a'} a',
\end{align}
Recall that $\mb{D}(N')$ takes 
\begin{equation}
a' \mapsto |\tau_{\geq 1} \mcal{X}(N')| \ \sum_{\alpha'} \langle a', \alpha' \rangle_N\  \alpha'.
\end{equation}
Thus the composition $\mb{D}(N') \circ Z_{\mcal{X}}(M)$ sends 
\begin{align}
a &\mapsto \frac{|\tau_{\geq 1} \mcal{X}(M)|} {|\tau_{\geq 1} \mcal{X}(N')|} \sum_{b \rightarrow a} (|\tau_{\geq 1} \mcal{X}(N')|\ \sum_{\alpha'} \langle q^*b, \alpha' \rangle_{N'} \ \alpha')\\
 &= |\tau_{\geq 1} \mcal{X}(M)| \sum_{b \rightarrow a} \sum_{\alpha'} \langle q^*b, \alpha' \rangle_{N'} \ \alpha'.
\end{align}

For $Z_{\dhcal{X}}(M) \circ \mbD(N)$, first $\mbD(N)$ sends:
\begin{equation}
    a \mapsto |\tau_{\geq 1}\mcal{X}(N)| \sum_{\alpha} \langle a, \alpha \rangle_N \ \alpha.
\end{equation}
$Z_{\dhcal{X}}(M)$ takes 
\begin{equation}
    \alpha \mapsto
    \frac{|\tau_{\geq 1} \dhcal{X}(M) |} {|\tau_{\geq 1} \dhcal{X}(N')|} \sum_{\beta, \beta \rightarrow \alpha} \hat{q}^* \beta.
\end{equation}
Thus the composition $Z_{\dhcal{X}}(M) \circ \mbD(N)$ is
\begin{align}
a \mapsto  |\tau_{\geq 1} \mcal{X}(N)| \frac{|\tau_{\geq 1} \dhcal{X}(M) |} {|\tau_{\geq 1} \dhcal{X}(N')|} 
\sum_{\alpha'} \sum_{\beta \rightarrow \alpha'} \langle a, \hat{p}^*\beta \rangle_N \ \alpha'.    
\end{align}

We are reduced to showing the following lemma:
\begin{lemma}
For every $a$ and $\alpha'$, $\sum_{b \rightarrow a} \langle q^*b, \alpha' \rangle_{N'}$ and $ \sum_{\beta \rightarrow \alpha'} \langle a, \hat{p}^*\beta \rangle_N$ differ a nonzero constant multiplicative $C$ that doesn't depend on $a$ or $\alpha'$.
\end{lemma}
\begin{proof}

Note that if $a$ has no preimage $b \mapsto a$. Then 
\begin{equation}
\sum_{b \rightarrow a} \langle q^*b, \alpha' \rangle_{N'} = 0.
\end{equation}
In this case, Lemma \ref{lemma_3} (stated and proven below) precise says that 
\begin{equation}
 \sum_{\beta \rightarrow \alpha'} \langle a, \hat{p}^*\beta \rangle_N = 0.   
\end{equation}
Similarly, if $\alpha'$ has no preimage $\beta \mapsto \alpha'$, then both sides are also zero. Thus we are reduced to the case that $a$ lies in the image of 
\begin{equation}
   p^*: \mcal{X}^0(M) \rightarrow \mcal{X}^0(N)
\end{equation}
and $\alpha'$ lies in the image of
\begin{equation}
  \hat{q}^*: \hcal{X}^{d-1}(M) \rightarrow \hcal{X}^{d-1}(N).  
\end{equation}

There are 
\begin{equation}
 |kp| \coloneqq |ker (p^*: \mcal{X}^0(M) \rightarrow \mcal{X}^0(N))|    
\end{equation}
many preimage of $a$. Similarly, there are 
\begin{equation}
|kq| \coloneqq |ker (\hat{q}^*: \dhcal{X}^0(M) \rightarrow \dhcal{X}^0(N'))|    
\end{equation}
many preimages of $\alpha'$.

On one side, we have 
\begin{gather}
\sum_{b \rightarrow a} \langle q^*(b), \alpha' \rangle_{N'} \\ 
= |kq|^{-1} \sum_{b \rightarrow a} \sum_{\beta \rightarrow \alpha'} \langle q^*(b), q^*\beta \rangle_{N'} \\
= |kq|^{-1} \sum_{b \rightarrow a} \sum_{\beta \rightarrow \alpha'} \langle p^*(b), p^*\beta \rangle_{N}.
\end{gather}
The last equation is by Lemma \ref{lemma_4}.
On the other side, we have 
\begin{gather}
\sum_{\beta \rightarrow \alpha'} \langle a, p^*\beta \rangle_N \\
= |kp|^{-1} \sum_{b \rightarrow a} \sum_{\beta \rightarrow \alpha'} \langle p^*(b), p^*\beta \rangle_{N}.
\end{gather}
We see that they differ by a constant $C = |kp|/|kq|$.
\end{proof}

Therefore 
\begin{equation}
\mbD(N') \circ Z_{\mcal{X}}(M) = \lambda(M)\  Z_{\dhcal{X}}(M) \circ \mbD(N),   
\end{equation}
with
\begin{equation}
\begin{aligned}
\lambda(M) = \frac{|\tau_{\geq 1} \mcal{X}(M)|}{|\tau_{\geq 1} \mcal{X}(N)|} \frac{|\tau_{\geq 1} \dhcal{X}(N')|}{|\tau_{\geq 1} \dhcal{X}(M)|} \frac{|kp|}{|kq|}.
\end{aligned} 
\end{equation}
\end{proof}

Now we need to prove Lemma \ref{lemma_3} and Lemma \ref{lemma_4} used above. 
We need the following lemma:
\begin{lemma}\label{lemma_5}
The natural maps 
\begin{equation}
 f: \hcal{X}_{1}(M, \dd M) \rightarrow \hcal{X}_0(N) \rightarrow \hcal{X}_0(M)
\end{equation}
and 
\begin{equation}
g: \hcal{X}_{1}(M, \dd M) \rightarrow \hcal{X}_0(N') \rightarrow \hcal{X}_0(M)
\end{equation}
are inverses to each other. That is, $f + g = 0$.
\end{lemma}

\begin{proof}
Consider the triple $\dd M \rightarrow M \rightarrow (M, \dd M)$, where $(M, \dd M)$ represents the cofiber. We have a long exact sequence
\begin{equation}
    \cdots \rightarrow \hcal{X}_{1}(M, \dd M) \rightarrow \hcal{X}_0(\dd M) \rightarrow \hcal{X}_0(M) \rightarrow \cdots.
\end{equation}
In particular, this means that the composition
\begin{equation}
    h: \hcal{X}_{1}(M, \dd M) \rightarrow \hcal{X}_0(\dd M) \rightarrow \hcal{X}_0(M)
\end{equation}
is the zero homomorphism $h = 0$.
As $\dd M = N \sqcup N'$, we see that $h = f + g$.
\end{proof}

We first prove Lemma \ref{lemma_4}. Because it might have independence interest, we recall the notations: 
we have $M: N \rightarrow N'$ a bordism between $N$ and $N'$, with the inclusion maps $p: N \hookrightarrow M$ and $q: N' \hookrightarrow M$. We have pullback maps 
\begin{equation}
p^*: \mcal{X}^0(M) \rightarrow \mcal{X}^0(N),\ \ q^*: \mcal{X}^0 (M) \rightarrow \mcal{X}^0 (N').
\end{equation}
Similarly we have
\begin{equation}
\hat{p}^*: \hcal{X}^{d-1}(M) \rightarrow \hcal{X}^{d-1}(N), \ \ \hat{q}^*: \hcal{X}^{d-1}(M) \rightarrow \hcal{X}^{d-1}(N').
\end{equation} 
We will denote elements of $\mcal{X}^0(M)$ as $b$ and $\hcal{X}^{d-1}(M)$ as $\beta$. 
Given $b$ and $\beta$, we have the two pairings (Equation \ref{pairing}): 
\begin{equation}
\langle p^*b, \hat{p}^*\beta \rangle_N, \ \ \langle q^*b, \hat{q}^*\beta \rangle_{N'}.
\end{equation}
Here's the lemma that we need to show:
\begin{lemma}\label{lemma_4}
 $\langle p^*b, \hat{p}^*\beta \rangle_N =\langle q^*b, \hat{q}^*\beta \rangle_{N'}$.
\end{lemma}

\begin{proof}
Recall that the orientation class $[M]$ restricts to $[N]$ on $N$ and $-[N']$ on $N'$. 
We will first consider $\langle p^*b, \hat{p}^*\beta \rangle_N$.
By \Poincare duality (Theorem \ref{poincare}), there is an isomorphism of long exact sequences:
\begin{equation}
\begin{tikzcd}
\cdots \arrow{r} \arrow{d}
    & \hcal{X}^{d-1}(M) \arrow{d}{\int_{[M]}}  \arrow{r}{\hat{p}^*}
        & \hcal{X}^{d-1}(N) \arrow{d}{\int_{[N]}} \arrow{r}
            & \hcal{X}^d(M,N) \arrow{d} \arrow{r}
                & \cdots \arrow{d}\\
\cdots \arrow{r}
    & \hcal{X}_{1}(M, \dd M) \arrow{r}{\hat{\mu}_*}
        & \hcal{X}_0(N) \arrow{r}
            & \hcal{X}_0(M, N') \arrow{r}
                & \cdots
\end{tikzcd}
\end{equation}
By definition of $\langle -, - \rangle_N$,  we have 
\begin{align}
\langle p^*b, \hat{p}^*\beta \rangle_N &= ev_N(p^*b, \int_{[N]} \hat{p}^* \beta)\\
                                      &= ev_N(p^*b, \hat{\mu}_* \int_{[M]} \beta) \label{6.1.1}
\end{align}
Now consider the long exact sequence:
\begin{align}
    ...\rightarrow \mcal{X}^0(M) \xrightarrow[]{p^*} \mcal{X}^0(N) \rightarrow \mcal{X}^1(M,N) \rightarrow ...
\end{align}
By Brown-Comenetz duality (Corollary \ref{sec5-theorem}), taking Pontryagin dual term-wise is isormorphic to the following long exact sequence
\begin{align}
    ... \leftarrow \hcal{X}_0(M) \xleftarrow[]{\hat{p}_*} \hcal{X}_0(N) \leftarrow \hcal{X}_{1}(M,N) \leftarrow ...
\end{align}
The dual long exact sequences are connected by the ``projection formula": given $b \in \mcal{X}^0(M)$ and $\gamma \in \hcal{X}_0(N)$, then 
\begin{equation}
    ev_N(p^*b, \gamma) = ev_M(b, \hat{p}_* \gamma).
\end{equation}
Put it together with Equation \ref{6.1.1}:
\begin{align}
\langle p^*b, \hat{p}^*\beta \rangle_N &=  ev_N(p^*b, \hat{\mu}_* \int_{[M]} \beta)\\
&=  ev_M(b, \hat{p}_* \circ \hat{\mu}_* \int_{[M]} \beta)
\end{align}

The same argument shows that 
\begin{equation}
\langle q^*b, \hat{q}^*\beta \rangle_{N'} = ev_M(b, -\hat{q}_* \circ \hat{\nu}_* \int_{[M]} \beta)
\end{equation}
The minus sign comes from the fact that $[M]$ restricts to $-[N']$.

Note that the map
\begin{equation}
\hat{p}_* \circ \hat{\mu}_* : \hcal{X}_{1}(M, \dd M) \rightarrow \hcal{X}_{0}(M)
\end{equation}
is precisely the map $f$ in Lemma \ref{lemma_5}. Similarly, $\hat{q}_* \circ \hat{\nu}_* = g$. By Lemma \ref{lemma_5}, we see that 
\begin{equation}
   \hat{p}_* \circ \hat{\mu}_* \int_{[M]} \beta =   -\hat{q}_* \circ \hat{\nu}_* \int_{[M]} \beta,
\end{equation}
therefore 
\begin{equation}
 \langle p^*b, \hat{p}^*\beta \rangle_N =\langle q^*b, \hat{q}^*\beta \rangle_{N'}.
 \end{equation}
 \end{proof}

\begin{remark}
Heuristically, since the orientation class $[M]$ for $M$ is a homotopy from $p_*[N]$ to $q_*[N] \in \mcal{R}_{d-1}(M)$, therefore
\begin{equation}
\langle p^*b, \hat{p}^*\beta \rangle_N \approx \langle b, \beta \rangle_{p_*[N]} \approx \langle b, \beta \rangle_{q_*[N']} \approx \langle q^*b, \hat{q}^*\beta \rangle_{N'}.
\end{equation}
\end{remark}
Now we proof the following lemma:
\begin{lemma}\label{lemma_3}
Let $a \in \mcal{X}^0(N)$ and $\alpha' \in \hcal{X}^{d-1}(N')$. If $a$ is not in the image of $p^*: \mcal{X}^0(M) \rightarrow \mcal{X}^0(N)$, then 
\begin{equation}
\sum_{\beta \rightarrow \alpha'} \langle a, \hat{p}^* \beta \rangle_N = 0,
\end{equation}
where $\beta$ sums over $\hcal{X}^{d-1}(M)$.
\end{lemma}

\begin{proof}
If $\alpha'$ has no preimage in 
\begin{equation}
\hat{q}^*: \hcal{X}^{d-1}(M) \rightarrow \hcal{X}^{d-1}(N'),    
\end{equation}
then the sum is trivially 0. If $\alpha'$ has a preimage, say $\beta_\alpha'$. Then all other preimages of $\alpha$ are of the form $\beta_\alpha' + \beta_0$, where $\beta_0 \in ker(\hat{q}^*)$. Thus 
\begin{align}
    \sum_{\beta \rightarrow \alpha'} \langle a, \hat{p}^* \beta \rangle_N &= 
    \sum_{\beta_0 \in ker(\hat{q}^*)} \langle a, \hat{p}^* (\beta_\alpha' + \beta_0) \rangle_N \\
    &=(\langle a, \hat{p}^* \beta_\alpha' \rangle_N)\ 
    \sum_{\beta_0 \in ker(\hat{q}^*)} \langle a, \hat{p}^* \beta_0 \rangle_N.
\end{align}
Therefore it suffices to show that
\begin{equation}
    \sum_{\beta_0 \in ker(\hat{q}^*)} \langle a, \hat{p}^* \beta_0 \rangle_N = 0,
\end{equation} 
i.e. in the case where $\alpha' = 0$.

\Poincare duality (Theorem \ref{poincare}) gives an isomorphism of long exact sequences:
\begin{equation}
\begin{tikzcd}
\cdots \arrow{r} \arrow{d}
    & \hcal{X}^{d-1}(M, N')\arrow{d}\arrow{r}
        & \hcal{X}^{d-1}(M) \arrow{d}{\int_{[M]}} \arrow{r}{\hat{q}^*} 
            & \hcal{X}^{d-1}(N') \arrow{d}{\int_{[N']}} \arrow{r}
                & \cdots \arrow{d}\\
\cdots \arrow{r}
    & \hcal{X}_1(M, N) \arrow{r}{\hat{\mu}_*}
        & \hcal{X}_{1}(M, \dd M) \arrow{r}{\hat{\nu}_*}
            & \hcal{X}_0(N') \arrow{r}
                & \cdots
\end{tikzcd}
\end{equation}
Note that $\hat{\mu}_*$ represents a different map from the proof of Lemma \ref{lemma_4}.

Under \Poincare duality, $ker(\hat{q}^*)$ corresponds to $ker(\hat{\nu}_*) = im(\hat{\mu}_*)$. Given $\beta_0 \in ker(\hat{q}^*)$ with
\begin{equation}
    \int_{[M]} \beta_0 = \hat{\mu}_* \gamma,\ \  \gamma \in \hcal{X}_1(M,N),
\end{equation}
By definition of $\langle -, - \rangle_N$, we have:
\begin{align}
    \langle a, \hat{p}^* \beta_0 \rangle_N &= ev_N(a, \hat{\lambda}_*\int_{[M]} \beta) \\
                                           &= ev_N(a, (\hat{\lambda}_* \circ \hat{\mu}_*) \gamma).\label{6.1.2}
\end{align}
$\hat{\lambda}_*$ is the canonical map $\mcal{X}_0(M) \rightarrow \mcal{X}_0(N)$.
The composition 
\begin{equation}
    \hat{\lambda}_* \circ \hat{\mu}_*: \hcal{X}_1(M,N) \rightarrow \hcal{X}_0(N)
\end{equation}
is the Pontryagin dual of the 
\begin{equation}
    \dd^*: \mcal{X}^0(N) \rightarrow \mcal{X}^1(M, N).
\end{equation}
Therefore by Equation \ref{6.1.2}
\begin{align}
    \langle a, \hat{p}^* \beta_0 \rangle_N &= ev_N(a, (\hat{\lambda}_* \circ \hat{\mu}_*) \gamma)\\
                                           &= ev_{(M,N)}(\dd^* a, \gamma).
\end{align}
Thus
\begin{align}
|ker (\hat{\mu}_*)| \  \sum_{\beta_0 \in ker(\hat{p}^*)} \langle a, \hat{p}^* \beta_0 \rangle_N = \sum_{\gamma} ev_{(M,N)}( \dd^* a, \gamma), \label{6.1.3}
\end{align}
where $\gamma$ sums over $\hcal{X}_1(M,N)$.
Now consider the the long exact sequence:
\begin{align}
\cdots \rightarrow \mcal{X}^0(M) \xrightarrow[]{p^*} \mcal{X}^0(N) \xrightarrow{\dd^*} \mcal{X}^1(M,N) \rightarrow \cdots
\end{align}
By hypothesis, $a$ is not in the image of $p^*:\mcal{X}^0(M) \rightarrow \mcal{X}^0(N)$, therefore 
\begin{equation}
\dd^* a \in \mcal{X}^1(M,N)     
\end{equation}
is not the identity element. 
Thus 
\begin{equation}
   ev_{(M,N)}(\dd^* a, -) : \hcal{X}_1(M,N) \rightarrow \cx 
\end{equation}
 is a nontrivial character on $\hcal{X}_1(M,N)$. As the sum over all elements of the group paired with a nontrivial character is 0, we see that 
 \begin{equation}
     \sum_{\gamma} ev_{(M,N)}( \dd^* a, \gamma) = 0.
 \end{equation}
By Equation \ref{6.1.3}, we get
\begin{equation}
    \sum_{\beta_0 \in ker(\hat{p}^*)} \langle a, \hat{p}^* \beta_0 \rangle_N = 0.
\end{equation}
\end{proof}

\subsection{Proof of Lemma 2}\label{sec6-2}

We continue the notations from last section \S \ref{sec6-1}.
Recall from last section we have 
\begin{equation}
\mbD(N') \circ Z_{\mcal{X}}(M) = \lambda(M)\  Z_{\dhcal{X}}(M) \circ \mbD(N)    
\end{equation}
with 
\begin{equation}
\begin{aligned}
\lambda(M) = \frac{|\tau_{\geq 1} \mcal{X}(M)|}{|\tau_{\geq 1} \mcal{X}(N)|} \frac{|\tau_{\geq 1} \dhcal{X}(N')|}{|\tau_{\geq 1} \dhcal{X}(M)|} \frac{|kp|}{|kq|}.
\end{aligned} 
\end{equation}
To finish the proof of the main theorem, we need the following lemma:
\begin{lemma}
$\lambda(M) = |\mcal{X}|^{\chi(M)-\chi(M)}.$
\end{lemma}
\begin{proof}
Recall that
\begin{equation}
\begin{aligned}
\lambda(M) = \frac{|\tau_{\geq 1} \mcal{X}(M)|}{|\tau_{\geq 1} \mcal{X}(N)|} \frac{|\tau_{\geq 1} \dhcal{X}(N')|}{|\tau_{\geq 1} \dhcal{X}(M)|} \frac{|kp|}{|kq|}.
\end{aligned} 
\end{equation}
First we will move everything in $\hcal{X}$ to $\mcal{X}$.

The first term is 
\begin{equation}
    |\uptau_{\geq 1} \dhcal{X}(N')|.
\end{equation}
By \Poincare duality (Theorem \ref{poincare}) we have 
\begin{equation}
\hcal{X}^*(N') \simeq \hcal{X}_{d-1-*}(N').     
\end{equation}
Therefore
\begin{equation}
|\hcal{X}^i(N')| = |\hcal{X}_{d-1-i}(N')| = |\mcal{X}^{d-1-i}(N')|.
\end{equation}
The cardinality of Pontryagin dual groups are equal, thus
\begin{align}
|\uptau_{\geq 1} \dhcal{X}(N')| &= \frac{|\hcal{X}^{d-3}(N')|}{|\hcal{X}^{d-2}(N')|} 
\frac{|\hcal{X}^{d-5}(N')|}{|\hcal{X}^{d-4}(N')|} \cdots \\
&= \frac{|\hcal{X}_2(N')|}{|\hcal{X}_1(N')|} \frac{|\hcal{X}_4(N')|}{|\hcal{X}_3(N')|} \cdots\\
&= \frac{|\mcal{X}^2(N')|}{|\mcal{X}^1(N')|} \frac{|\mcal{X}^4(N')|}{|\mcal{X}^3(N')|} \cdots\\
&= |\uptau_{\leq -1} \mcal{X}(N')|.
\end{align}
Next we will work on 
\begin{equation}
  |\uptau_{\geq 1} \dhcal{X}(M)|^{-1}.  
\end{equation}
Similar to above, we have 
\begin{equation}
|\hcal{X}^i(M)| = |\hcal{X}_{d-i}(M, \dd M)| = |\mcal{X}^{d-i}(M, \dd M)|.    
\end{equation}
Therefore
\begin{align}
|\uptau_{\geq 1} \dhcal{X}(M)|^{-1} &= 
\frac{|\hcal{X}^{d-2}(M')|}{|\hcal{X}^{d-3}(M')|} \frac{|\hcal{X}^{d-4}(M')|}{|\hcal{X}^{d-5}(M')|} \cdots \\
&= \frac{|\hcal{X}_2(M, \dd M)|}{|\hcal{X}_3(M, \dd M)|} 
\frac{|\hcal{X}_4(M, \dd M)|}{|\hcal{X}_5(M, \dd M)|} \cdots \\
&= \frac{|\mcal{X}^2(M, \dd M)|}{|\mcal{X}^3(M, \dd M)|}
\frac{|\mcal{X}^4(M, \dd M)|}{|\mcal{X}^5(M, \dd M)|} \cdots \\
&= |\uptau_{\leq -1} \mcal{X}(N')|.
\end{align}
Lastly, we have
\begin{equation}
|kq| \coloneqq |ker (\hat{q}^*: \hcal{X}^{d-1}(M) \rightarrow \hcal{X}^{d-1}(N'))|.
\end{equation}
 By Poincare duality (Theorem \ref{poincare}): we have that an isomorphism of long exact sequences: 
\begin{equation}
\begin{tikzcd}
\cdots \arrow{r} \arrow{d}
    & \hcal{X}^*(M, N') \arrow{d}\arrow{r}
        & \hcal{X}^*(M) \arrow{d} \arrow{r}{q^*}
            & \hcal{X}^*(N') \arrow{d} \arrow{r}
                & \cdots \arrow{d}\\
\cdots \arrow{r}
    & \hcal{X}_{d-*}(M, N) \arrow{r}
        & \hcal{X}_{d-*}(M, \dd M) \arrow{r}
            & \hcal{X}_{d-1-*}(N') \arrow{r}
                & \cdots
\end{tikzcd}
\end{equation}
Thus 
\begin{align}
|kq| &= |ker(\hcal{X}_{1}(M, \dd M) \rightarrow \hcal{X}_0(N'))| \\ &= |im (\hcal{X}_{1}(M, N) \rightarrow \hcal{X}_{1}(M, \dd M)|    .
\end{align}
By Brown-Comenetz duality (Corollary \ref{sec5-theorem}), the long exact sequence 
\begin{equation}
\cdots \rightarrow \hcal{X}_{d-*}(M, N) \rightarrow \hcal{X}_{d-*}(M, \dd M) \rightarrow  \hcal{X}_{d-1-*}(N') \rightarrow \cdots
\end{equation}
is the Pontryagin dual of
\begin{equation}
\cdots \leftarrow \mcal{X}^{d-*}(M, N) \leftarrow \mcal{X}^{d-*}(M, \dd M) \leftarrow \mcal{X}^{d-1-*}(N') \leftarrow \cdots
\end{equation}
Thus
\begin{align}
|kq| &= |im (\hcal{X}_{1}(M, N) \rightarrow \hcal{X}_{1}(M, \dd M)|\\ &= |im (\mcal{X}^{1}(M, \dd M) \rightarrow \mcal{X}^{1}(M, N)|\\  &=  |ker (\mcal{X}^1(M, N) \rightarrow \mcal{X}^1(N'))|.    
\end{align}
To recap, we have
\begin{equation}
\begin{aligned}
\lambda(M) = \frac{|\tau_{\geq 1} \mcal{X}(M)|}{|\tau_{\geq 1} \mcal{X}(N)|}  \frac{|ker (\mcal{X}^0(M) \xrightarrow{p^*} \mcal{X}^0(N))|}{|ker (\mcal{X}^1(M, N) \rightarrow \mcal{X}^1(N'))|}\\
{|\uptau_{\leq -1} \mcal{X}(N')|}{|\uptau_{\leq -1} \mcal{X}(N')|}.
\end{aligned} 
\end{equation}

Now we will factor out
$|\mcal{X}|^{\chi(M)-\chi(M)}$ from $\lambda(M)$.
For any $\pi$-finite space $\mcal{Y}$, we have a fiber sequence
\begin{equation}
    \uptau_{\geq i}\mcal{Y} \rightarrow \mcal{Y} \rightarrow \uptau_{\leq i-1} \mcal{Y}
\end{equation}
of $\pi$-finite spaces. By Example \ref{sec1-example} we have 
\begin{equation}
    |\uptau_{\geq i}\mcal{Y}| \ |\uptau_{\leq i-1} \mcal{Y}| = |\mcal{Y}|.
\end{equation}
In our case,
\begin{equation}
|\uptau_{\geq 1} \mcal{X}(M)| = \frac{|\mcal{X}(M)|}{|\uptau_{\leq 0} \mcal{X}(M)|}.    
\end{equation}
Similarly, 
\begin{equation}
|\uptau_{\geq 1} \mcal{X}(N)|^{-1} = \frac{|\uptau_{\leq 0} \mcal{X}(N)|}{|\mcal{X}(N)|}.    
\end{equation}
By Proposition \ref{sec1-proposition2} we have 
\begin{equation}
|\mcal{X}(M)| = |\mcal{X}|^{\chi(M)}
\end{equation}
and
\begin{equation}
|\mcal{X}(N)| = |\mcal{X}|^{\chi(N)}.
\end{equation}
Putting it all together, we see that 
\begin{equation}
\lambda (M) = \lambda'(M) \  |\mcal{X}|^{\chi(M)-\chi(M)},    
\end{equation}
 where 
 \begin{equation}
 \begin{gathered}
\lambda'(M) = \frac{|\uptau_{\leq 0} \mcal{X}(N)|}{|\uptau_{\leq 0} \mcal{X}(M)|} |\uptau_{\leq -1} \mcal{X}(N')|\ |\uptau_{\leq -2} \mcal{X}(M, \dd M)|\\
\frac{|ker (\mcal{X}^0(M) \xrightarrow{p^*} \mcal{X}^0(N))|}{|ker (\mcal{X}^1(M, N) \rightarrow \mcal{X}^1(N'))|}
\end{gathered}
 \end{equation}
It remains to show that $\lambda'(M) = 1$.

As 
\begin{equation}
\dd M = N \sqcup N',
\end{equation}
we have 
\begin{equation}
|\mcal{X}^*(\dd M)| = |\mcal{X}^*(N)|\ |\mcal{X}^*(N')|.    
\end{equation}
Therefore
\begin{equation}
|\uptau_{\leq 0} \mcal{X}(N)|\  |\uptau_{\leq -1} \mcal{X}(N')| = |\mcal{X}^0(N)|\  |\uptau_{\leq -1} \mcal{X}(\dd M)|.
\end{equation}
Now consider the exact sequences 
\begin{equation}
0 \rightarrow ker\  p^* \rightarrow \mcal{X}^0(M) \xrightarrow[]{p^*} \mcal{X}^0(N) \rightarrow coker\  p^* \rightarrow 0,
\end{equation}
We see that the terms 
\begin{equation}
|ker\ p^*|\  |\uptau_{\leq 0} \mcal{X}(M)|^{-1}\  |\uptau_{\leq 0} \mcal{X}(N)| = |coker \  p^*|. 
\end{equation} 
Lastly, we rewrite 
\begin{equation}
|coker \ p^*| = |ker\  \mcal{X}^1(M,N) \rightarrow \mcal{X}^1(M)|.     
\end{equation}

Thus 
\begin{equation}
\begin{gathered}
\lambda'(M) = \frac{|\uptau_{\leq -1} \mcal{X}(\dd M)|}{|\uptau_{\leq -1} \mcal{X}(M)|}\ |\uptau_{\leq -2} \mcal{X}(M, \dd M)| \\ 
\frac{|ker (\mcal{X}^1(M,N) \rightarrow \mcal{X}^1(M))|}{|ker (\mcal{X}^1(M, N) \rightarrow \mcal{X}^1(N'))|}.
\end{gathered}
\end{equation}
We claim that 
\begin{equation}
\frac{|ker (\mcal{X}^1(M,N) \rightarrow \mcal{X}^1(M))|}{|ker (\mcal{X}^1(M, N) \rightarrow \mcal{X}^1(N'))|}= |ker (\mcal{X}^1(M) \rightarrow \mcal{X}^1(\dd M))|^{-1}.
\end{equation}
First notice that the canonical map 
\begin{equation}
\mcal{X}^1(M,N) \rightarrow \mcal{X}^1(N) = 0,
\end{equation}
therefore 
\begin{gather}
|ker (\mcal{X}^1(M,N) \rightarrow \mcal{X}^1(N'))| = |ker (\mcal{X}^1(M,N) \rightarrow \mcal{X}^1(\dd M))|.
\end{gather}
Note that
\begin{equation}
\mcal{X}^1(M,N) \rightarrow \mcal{X}^1(\dd M)    
\end{equation}
is the composition of the two terms
\begin{equation}
(\mcal{X}^1(M) \rightarrow \mcal{X}^1(\dd M)) \circ   (\mcal{X}^1(M, N) \rightarrow \mcal{X}^1(M))
\end{equation}
on the RHS. Therefore we are trying to show this:
\begin{align}
|ker (\mcal{X}^1(M,N) \rightarrow \mcal{X}^1(\dd M))
 = |ker (\mcal{X}^1(M,N) \rightarrow \mcal{X}^1(M))| \\
 |ker (\mcal{X}^1(M) \rightarrow \mcal{X}^1(\dd M))|.
\end{align}

We have the following algebraic fact: given 
\begin{equation}
f: A \rightarrow B,\   g: B \rightarrow C  
\end{equation}
then 
\begin{equation}
|ker (g \circ f)| = |ker f|\ |ker g|    
\end{equation}
iff 
\begin{equation}
ker(g) \subset im(f).    
\end{equation}
In our case, if an element $a \in \mcal{X}^1(M)$ maps to $0$ in $\mcal{X}^1(\dd M)$, then it maps to $0$ in $\mcal{X}^1(N)$. Since 
\begin{equation*}
    \mcal{X}^1(M, N) \rightarrow \mcal{X}^1(M) \rightarrow \mcal{X}^1(N)
\end{equation*}
is a part of a long exact sequence, it is exact at $\mcal{X}^1(M)$. That means that there exists $b \in \mcal{X}^1(M,N)$ which maps to $a$. 
Thus we satisfy the algebraic condition, and we have 
\begin{equation}
\frac{|ker (\mcal{X}^1(M,N) \rightarrow \mcal{X}^1(M))|}{|ker (\mcal{X}^1(M, N) \rightarrow \mcal{X}^1(N'))|}= |ker (\mcal{X}^1(M) \rightarrow \mcal{X}^1(\dd M))|^{-1}.
\end{equation}
So
\begin{equation}
\lambda'(M) = \frac{|\uptau_{\leq -1} \mcal{X}(\dd M)|}{|\uptau_{\leq -1} \mcal{X}(M)|}
\frac{|\uptau_{\leq -2} \mcal{X}(M, \dd M)|}{|ker (\mcal{X}^1(M) \rightarrow \mcal{X}^1(\dd M))|}.
\end{equation}
Finally, consider the following long exact sequence:
\begin{gather} \label{final_long}
0 \rightarrow ker (\mcal{X}^1(M) \rightarrow \mcal{X}^1(\dd M)) \rightarrow  \mcal{X}^1(M)
\rightarrow \mcal{X}^1(\dd M) \\ \rightarrow \mcal{X}^2(M, \dd M) \rightarrow \mcal{X}^2(M) \rightarrow \mcal{X}^2(\dd M) \rightarrow \cdots.
\end{gather}
By Lemma \ref{sec1-lemma} the alternating size of the finite abelian groups in a long exact sequence is 1. The alternating size of the long exact sequence \ref{final_long} above is precisely $\lambda'(M)$, thus \begin{equation}
\lambda'(M) = 1.
\end{equation}
\end{proof}